\documentclass[titlepage,twoside,11pt,letter]{article}
\usepackage[margin=1.2in]{geometry}

\usepackage{times}
\usepackage{amsmath}
\usepackage{amsfonts}
\usepackage{amssymb}
\usepackage{graphicx}
\usepackage{mathrsfs}
\usepackage{tabularx}
\usepackage{caption}
\usepackage{bbold}
\usepackage{color}
\usepackage{fancybox}
\usepackage{verbatim}
\usepackage{tikz}
\usetikzlibrary{arrows,automata}
\usetikzlibrary{trees}

\def\bR{\mathbb{R}}

\def\b1{\mathbb{1}}

\def\eE{\mathsf{E}}

\def\eP{\mathsf{P}}

\def\tT{\mathtt{T}}

\def\cB{\mathcal{B}}
\def\cC{\mathcal{C}}

\def\cP{{\mathcal{P}}}

\def\cU{{\mathcal{U}}}

\def\sF{{\mathscr{F}}}

\def\e1{\mathsf{1}}

\def\ofs{(s)}

\def\oft{(t)}

\def\ofv{(v)}
\def\ofw{(w)}
\def\ofx{(x)}
\def\ofy{(y)}

\def\of0{(0)}

\def\d{\partial}
\def\tr{\mathrm{tr}}
\def\id{\mathrm{I}}

\def\bf0{\mathbf{0}}
\def\cp1{\mathbb{CP}^1}
\def\matn{\mathrm{Mat}_n(\bR)}

\def\arg{\mathrm{arg}}

\newtheorem{theorem}{Theorem}[section]

\newtheorem{proposition}[theorem]{Proposition}

\newenvironment{proof}[1][Proof]{\noindent\textbf{#1:} }{\ \rule{0.5em}{0.5em}}

\begin{document}

\title{\textbf{Semiclassical approximation in stochastic optimal control\\
I. Portfolio construction problem}}
\author{\textbf{Sakda Chaiworawitkul}\\
JPMorgan Chase\\
New York, NY 10179\\
USA
\and \textbf{Patrick S. Hagan}\\
Mathematics Institute\\
24-29 St Giles\\
Oxford University\\
Oxford, OX1 3LB\\
UK
\and \textbf{Andrew Lesniewski}\\
Department of Mathematics\\
Baruch College\\
One Bernard Baruch Way\\
New York, NY 10010\\
USA}
\date{First draft: December 3, 2013\\
This draft: \today}
\maketitle

\begin{abstract}
This is the first in a series of papers in which we study an efficient approximation scheme for solving the Hamilton-Jacobi-Bellman equation for multi-dimensional problems in stochastic control theory. The method is a combination of a WKB style asymptotic expansion of the value function, which reduces the second order HJB partial differential equation to a hierarchy of first order PDEs, followed by a numerical algorithm to solve the first few of the resulting first order PDEs. This method is applicable to stochastic systems with a relatively large number of degrees of freedom, and does not seem to suffer from the curse of dimensionality. Computer code implementation of the method using modest computational resources runs essentially in real time. We apply the method to solve a general portfolio construction problem.
\end{abstract}

\tableofcontents

\section{\label{sec:Introduction}Introduction}

The stochastic Hamilton-Jacobi-Bellman (HJB) partial differential equation is the cornerstone of stochastic optimal control theory (\cite{FS92}, \cite{YZ99}, \cite{P09}). Its solution, the value function, contains the information needed to determine the optimal policy governing the underlying dynamic optimization problem. Analytic closed form solutions to the HJB equation are notoriously difficult to obtain, and they are limited to problems where the underlying state dynamics has a simple form. Typically, these solutions are only available for systems with one degree of freedom.

A variety of numerical approaches to stochastic optimal control have been studied. An approach based on the Markov chain approximation is developed in \cite{KD01}. This approach avoids referring to the HJB equation altogether, and is, instead, based on a suitable discretization of the underlying stochastic process. Other recent approaches, such as \cite{FLO7}, \cite{KLP13}, and \cite{AK14}, rely on ingenious discretization schemes of the HJB equation. These numerical methods are generally limited to systems with low numbers of degrees of freedom, as they are susceptible to the ``curse of dimensionality''.

In this paper, we present a methodology for effectively solving a class of stochastic HJB equations for systems with $n$ degrees of freedom, where $n$ is a moderately large number ($\lessapprox 200$). The solution methodology is based on an analytic approximation to the full HJB equation which reduces it to an infinite hierarchy of first order partial differential equations. This is accomplished by means of an asymptotic expansion analogous to the Wentzel-Kramers-Brillouin (WKB) method used in quantum mechanics, optics, quantitative finance, and other fields of applied science, see e.g. \cite{BO99}, \cite{KC85}. The first in the hierarchy of equations is the classical Hamilton-Jacobi (HJ) equation which is analogous to the equation describing the motion of a particle on a Riemannian manifold\footnote{The language of Riemannian geometry provides a natural, albeit somewhat technical, framework for WKB expansion of the HJB equation, and we intend to discuss it in a separate paper.} subject to external forces. Its structure is somewhat less complicated than that of the full HJB equation, and its properties have been well understood. The solution to this equation is in essence the most likely trajectory for the optimal control of the stochastic system. Similar ideas, within a completely different setup have been pursued in \cite{T11} and \cite{HDM14}. The remaining equations are linear first order PDEs, with progressively more complex structure of coefficient functions.

The approximate character of the solution of the HJB equation that we discuss is twofold. Firstly, we solve the Hamilton-Jacobi equation and the first of the linear PDEs in the hierarchy only. The WKB expansion is asymptotic, and the expectation is that these two equations capture the nature of the actual solution close enough. The remaining members of the hierarchy are neglected as they are believed that they contain information which does not significantly affect the shape of the solution. We refer to this approximation as the semiclassical (or eikonal) approximation in analogy with a similar approximation in physics. Interestingly, there is a class of non-trivial stochastic optimal control problems for which the semiclassical approximation produces the actual exact solutions. Two examples of such problems are discussed in the paper.

Secondly, the solutions to the two leading order PDEs are constructed through numerical approximations. The key element of the numerical algorithm is a suitable symplectic method of numerical integration of Hamilton's canonical equations, which are the characteristic equations of the HJ equation. Here, we use the powerful St\"ormer-Verlet (or leapfrog) method \cite{HLW03}, \cite{LR04} to construct numerically the characteristics. Furthermore, we use a Newton-type search method in order to construct the numerical solution to the HJ equation out of the characteristics. This method uses a system of variational equations associated with Hamilton's equations.

This work has been motivated by our study of a stochastic extension of the continuous time version of the Markowitz mean variance portfolio optimization. The methodology developed here should, however, provide a practical method for implementation of the resulting portfolio construction. We believe, however, that the method is of broader interest and can be applied to a class of stochastic optimization problems outside of portfolio construction theory.

\section{\label{sec:hjbEq}Portfolio construction problem and the HJB equation}

We assume that the underlying source of stochasticity is a standard $p$-dimensional Wiener process $Z\oft\in\bR^p$ with independent components,
\begin{equation}
\eE[dZ\oft dZ\oft^\tT]=\id dt.
\end{equation}
Here, $\id$ denotes the $p\times p$ identity matrix. We let $(\Omega,(\sF)_{t\geq 0},\eP)$ denote the filtered probability space, which is associated with the Wiener process $Z$.

We formulate the portfolio construction problem as the following stochastic control problem. We consider a controlled stochastic dynamical system whose states are described by a multi-dimensional diffusion process $(X\oft,W\oft)$, which takes values in $\cU\times\bR$, where $\cU\subset\bR^n$  is an open set. The components $X^i$, $i=1,\ldots,n$, of $X$ represent the prices of the individual assets in the portfolio, and $W$ is total value of the portfolio. We assume that $n\leq p$. The allocations of each of the assets in the portfolio are given by an $(\sF)_{t\geq 0}$-adapted process $\varphi\oft\in\bR^n$.

The dynamics of $(X,W)$ is given by the system of stochastic differential equations:
\begin{equation}\label{eq:xDyn}
\begin{split}
dX\oft&=a(X\oft)dt+b(X\oft)dZ\oft,\\
X\of0&=X_0.
\end{split}
\end{equation}
The drift and diffusion coefficients $\cU\ni x\to a(x)\in\bR^n$ and $\cU\ni x\to b(x)\in\mathrm{Mat}_{n,p}(\bR)$, respectively, satisfy the usual H\"older and quadratic growth conditions, which guarantee the existence and uniqueness of a strong solution to this system. Note that we are not requiring the presence of a riskless asset in the portfolio: such assumption is unrealistic and unnecessary. If one wishes to consider a riskless asset, it is sufficient to take a suitable limit of the relevant components of $a$ and $b$. The process $W$ is given by
\begin{equation}\label{eq:yDyn1}
\begin{split}
dW\oft&=\varphi\oft^\tT dX\oft,\\
W\of0&=W_0.
\end{split}
\end{equation}
Explicitly, equation \eqref{eq:yDyn1} reads:
\begin{equation}\label{eq:yDyn}
dW\oft=\varphi\oft^\tT a(X\oft)dt+\varphi\oft^\tT b(X\oft)dZ\oft.
\end{equation}
We refer to the process $W$ as the investor's wealth process.

We assume that the investor has a finite time horizon $T$ and the utility function $U$. We shall assume that $U$ is a member of the HARA family of utility functions, see Appendix \ref{sec:UtilityFunctions} for the their definition and summary of properties. The investor's objective is to maximize the expected utility of his wealth at time $T$. We are thus led to the following cost functional:
\begin{equation}
J[\varphi]=\eE\big[U(W(T))\big],
\end{equation}
which represents the investor's objective function.

Let
\begin{equation}\label{eq:covDef}
\cC\ofx=b\ofx^\tT b\ofx
\end{equation}
denote the instantaneous covariance matrix of the price processes. For technical reasons, we shall make the following additional assumptions on the functions $a:\cU\to\bR^n$ and $b:\cU\to\mathrm{Mat}_{n,p}(\bR)$:
\begin{itemize}
\item[(A1)]{The functions $a\ofx$ and $b\ofx$ are three times continuously differentiable for all $x\in\cU$.}
\item[(A2)]{The matrix $\cC\ofx$ is positive definite for all $x\in\cU$.}
\end{itemize}
In particular, the function $x\to\cC\ofx^{-1}$ is three times continuously differentiable.

Our goal thus is to find the optimal policy $\varphi^\ast$ which maximizes the expected utility if the terminal value of $W$. In other words, we are seeking the $\varphi^\ast$ such that
\begin{equation}
\varphi^\ast=\mathop{\arg\sup}_\varphi\;\eE\big[U(W(T))\big].
\label{eq:OptimizationProblem}
\end{equation}

We solve this optimization problem by invoking stochastic dynamic programming, see eg. \cite{FS92}, \cite{YZ99} or \cite{P09}. The key element of this approach is the value function $J(t,x,w)$. It is determined by two requirements:
\begin{itemize}
\item[(B1)]{it satisfies Bellman's principle of optimality,
\begin{equation}
J(t,X\oft,W\oft)=\sup_{\varphi}\;\eE\big[J(t+dt,X(t+dt),W(t+dt))|\sF_t\big],
\end{equation}
for all $0\leq t<T$, and}
\item[(B2)]{it satisfies the terminal condition,
\begin{equation}
J(T,X(T),W(T))=U(W(T)).
\end{equation}}
\end{itemize}
These conditions lead to the following nonlinear PDE for the value function,
\begin{equation}
\dot{J}+\sup_{\varphi}\,\big\{a^\tT\nabla_x J+\varphi^\tT a\nabla_w J+\frac{1}{2}\,\tr(\cC\nabla^2_x J)+\varphi^\tT\cC\nabla^2_{xw} J+\frac{1}{2}\,\varphi^\tT\cC\varphi\nabla^2_w J\big\}=0,
\end{equation}
namely the stochastic Hamilton-Jacobi-Bellman equation, subject to the terminal condition
\begin{equation}
J(T,x,w)=U(w).
\label{eq:defValueFunctionTermimal}
\end{equation}

In order to solve the HJB equation, we choose $\varphi=\varphi^\ast$ which formally maximizes the expression inside the curly parentheses above. In other words, $\varphi^*$ satisfies
\begin{equation*}
(\nabla^2_w J)\cC\varphi+a\nabla_w J+\cC\nabla^2_{xw} J=0.
\end{equation*}
This leads to the following condition:
\begin{equation}
\varphi^\ast=-\frac{\nabla^2_{xw} J}{\nabla^2_w J}-\frac{\nabla_w J}{\nabla ^2_w J}\,\cC^{-1} a,
\end{equation}
known as the first order condition. Substituting $\varphi^*$ back to the HJB equation yields
\begin{equation}\label{eq:hjbRed}
\dot{J}+a^\tT\nabla_x J+\frac{1}{2}\,\tr(\cC\nabla_x^2 J)-\frac{1}{2\nabla^2_w J}\,(\nabla^2_{xw}J+\cC^{-1}a\nabla_w J)^\tT\cC(\nabla^2_{xw}J+\cC^{-1} a\nabla_w J)=0.
\end{equation}
We solve this equation by means of the following Ansatz:
\begin{equation}
J(t,x,w)=\Gamma(t,x)U\ofw.
\end{equation}
Using Proposition \ref{thm:haraProps} we find that $\Gamma$ satisfies the following non-linear PDE:
\begin{equation}\label{eq:hjbGamma}
\dot{\Gamma}+a^\tT\nabla_x \Gamma+\frac{1}{2}\,\tr(\cC\nabla^2\Gamma)+\frac{\kappa}{2}\,(\nabla\log\Gamma+\cC^{-1} a)^\tT \cC(\nabla\log\Gamma+\cC^{-1} a)\Gamma=0,
\end{equation}
subject to the terminal condition
\begin{equation}
\Gamma(T,x)=1.
\end{equation}
The constant $\kappa$ depends only on the utility function and is given explicitly by \eqref{eq:kappaDef}. Since it will lead to no ambiguity, we have suppressed the subscript $x$ in the derivatives with respect to $x$.

Note that the optimal control $\varphi^*$ has the following expression in terms of $\Gamma$ and $U$:
\begin{equation}\label{eq:optCont}
\varphi^\ast=\frac{1}{A_U\ofw}\big(\nabla\log\Gamma+\cC^{-1} a\big),
\end{equation}
where $A_U\ofw$ is the absolute risk aversion coefficient of the utility $U$.

\section{\label{sec:wkbExp}WKB expansion of the HJB equation}

We shall write down the solution to equation \eqref{eq:hjbGamma} in terms of an asymptotic expansion, namely the WKB expansion. The first few terms of this expansion yield an approximate solution which is sometimes referred as the semiclassical or eikonal approximation.

The WKB asymptotic expansion is based on the assumption that the covariance matrix $\cC$ is ``small'' in some sense. To this end we scale the covariance matrix,
\begin{equation}
\cC\to\varepsilon \cC,
\end{equation}
where $\varepsilon$ is a parameter used to keep track of the order of magnitude in terms of $\cC$. At the end of the calculation, $\varepsilon$ is set back to $1$. Then, equation \eqref{eq:hjbGamma} takes the form:
\begin{equation}\label{eq:hjbGammaEps}
\dot{\Gamma}+a^\tT\nabla \Gamma+\frac{\varepsilon}{2}\,\tr(\cC\nabla^2\Gamma)+\frac{\varepsilon\kappa}{2}\,(\nabla\log\Gamma+\varepsilon^{-1}\cC^{-1} a)^\tT \cC(\nabla\log\Gamma+\varepsilon^{-1}\cC^{-1} a)\Gamma=0.
\end{equation}
We seek a solution to the equation above in the form
\begin{equation}
\Gamma(t,x)=\exp\big(\tfrac{1}{\varepsilon}\,S(t,x)\big),
\end{equation}
where $S(t,x)$ has a finite limit as $\varepsilon\to 0$. Substituting this Ansatz into \eqref{eq:hjbGammaEps}, we find that the equation for $S$ reads
\begin{equation}\label{eq:hjbS}
\dot{S}+(\kappa+1)a^\tT\nabla S+\frac{\kappa+1}{2}\,(\nabla S)^\tT\cC\nabla S+\frac{\kappa}{2}\,a^\tT\cC^{-1} a
+\frac{\varepsilon}{2}\,\tr(\cC\nabla^2 S)=0.
\end{equation}
The optimal control expressed in terms of $S$ takes the following form:
\begin{equation}
\varphi^\ast=\frac{1}{A_U}\Big(\cC^{-1}a+\frac{1}{\varepsilon}\,\nabla S\Big).
\end{equation}

We assume that $S$ has an asymptotic expansion in powers of $\varepsilon$,
\begin{equation}
S(t,x)=S^0(t,x)+S^1(t,x)\varepsilon+S^2(t,x)\varepsilon^2+\mathrm{O}(\varepsilon^3).
\end{equation}
Substituting this expansion into equation \eqref{eq:hjbS} yields an infinite hierarchy of equations:
\begin{equation}\label{eq:wkbHierarchy}
\begin{split}
\dot{S}^0&+\frac{\kappa+1}{2}\,(\nabla S^0)^\tT\cC\nabla S^0+(\kappa+1)a^\tT\nabla S^0+\frac{\kappa}{2}\,a^\tT\cC^{-1} a=0,\\
\dot{S}^1&+(\kappa+1)(a+\nabla S^0)^\tT\cC\nabla S^1+\frac{1}{2}\,\tr(\cC\nabla^2 S^0)=0,\\
\dot{S}^2&+(\kappa+1)(a+\nabla S^0)^\tT\cC\nabla S^2+\frac{\kappa+1}{2}\,(\nabla S^1)^\tT\cC\nabla S^1+\frac{1}{2}\,\tr(\cC\nabla^2 S^1)=0,\\
&\ldots\,,\\
\end{split}
\end{equation}
where each of the $S^j$'s satisfies the terminal condition:
\begin{equation}\label{eq:bdCondS}
S^j(T,x)=0,\text{ for } j=0,1,2,\dots .
\end{equation}
The first of these equations is non-linear in $S^0$. Each of the subsequent equations is a linear PDE, with coefficients that depend on the solutions of the preceding equations.

We define the variables $p$ dual to $x$ by
\begin{equation}\label{eq:pdef}
p\triangleq\nabla S^0,
\end{equation}
and refer to $p$ as the canonical momenta conjugate with $x$. We can then write the first of the equations \eqref{eq:wkbHierarchy} as
\begin{equation}\label{eq:hjEqu}
\dot{S}^0+H(x,\nabla S^0)=0,
\end{equation}
where the Hamiltonian $H(x,p)$ is given by
\begin{equation}\label{eq:hamDef}
H(x,p)=\frac{1}{2\gamma}\,p^\tT\cC\ofx p+\frac{1}{\gamma}\,p^\tT a\ofx+V(x),
\end{equation}
where
\begin{equation}\label{eq:vDef}
V\ofx=\frac{\kappa}{2}\,a\ofx^\tT\cC\ofx^{-1} a\ofx.
\end{equation}
This non-linear PDE is the classical Hamilton-Jacobi equation, see e.g. \cite{CH53}, \cite{E92}. Its solution gives the leading order approximation to the solution of the stochastic Hamilton-Jacobi-Bellman equation. From the physics point of view, the Hamiltonian \eqref{eq:hamDef} describes the dynamics of a particle of mass $\gamma$ moving on a suitable Riemannian manifold in the potential $V\ofx$ and subject to an additional velocity dependent force\footnote{Alternatively, one can interpret it as a motion on $\cU$ in a magnetic field with potential $-a\ofx$ subject to the external potential $-\frac{1}{2}\,a\ofx^\tT\cC\ofx^{-1} a\ofx$.}. The solutions to the remaining linear equations in the hierarchy yield sub-leading ``stochastic'' corrections to the classical solution:
\begin{equation}\label{eq:wkbApprPhi}
\Gamma(t,x)=\exp\big(\tfrac{1}{\varepsilon}\,S^0(t,x)+S^1(t,x)\big)\big(1+O(\varepsilon)\big).
\end{equation}
This approximation is analogous to the eikonal approximation in classical optics or the semiclassical approximation in classical mechanics.

\section{\label{sec:solHier}Solving the WKB hierarchy}

We shall now describe the methodology for solving the WKB hierarchy \eqref{eq:wkbHierarchy}. Since each equation in the hierarchy is a first order PDE, the appropriate approach consists in applying the method of characteristics, see e.g. \cite{CH53} and \cite{E92}.

We begin by solving the Hamilton-Jacobi equation \eqref{eq:hjEqu}. To this end, we recall that its characteristic equations are given by:
\begin{equation}\label{eq:charEqHJ}
\begin{split}
\dot{x}\ofs&=\nabla_p H(x\ofs,p\ofs),\\
\dot{p}\ofs&=-\nabla_x H(x\ofs,p\ofs),\\
\dot{z}\ofs&=p\ofs^\tT\nabla_p H(x\ofs,p\ofs)-H(x\ofs,p\ofs),
\end{split}
\end{equation}
where $z\ofs=S^0(s,x\ofs)$. These equations are subject to the terminal condition:
\begin{equation}\label{eq:termValCharEqHJ}
\begin{split}
x(T)&=y\,,\\
p(T)&=0\,,\\
z(T)&=0,
\end{split}
\end{equation}
where the terminal conditions for $p$ and $z$ are consequences of \eqref{eq:bdCondS} and \eqref{eq:pdef}.

The first two of the characteristic equations \eqref{eq:charEqHJ} are canonical Hamilton's equations associated with the Hamiltonian $H$. Classic results of the theory of ordinary differential equations, see eg. \cite{CL55}, guarantee the existence and uniqueness of the solution to the above terminal value problem, at least for $T$ sufficiently small. Furthermore, the solution depends smoothly on the terminal value $y$.

In order to analyze Hamilton's equations, assume first that $\kappa\neq-1$. They read then:
\begin{equation}\label{eq:hamEqs1}
\begin{split}
\dot{x}&=\frac{1}{\gamma}\,(\cC\ofx p+a\ofx),\\
\dot{p}&=-\nabla_x\Big(\frac{1}{2\gamma}\,p^\tT\cC\ofx p+\frac{1}{\gamma}\,p^\tT a\ofx+V\ofx\Big),\\
\end{split}
\end{equation}
or, explicitly,
\begin{equation}\label{eq:hamEqs2}
\begin{split}
\dot{x}&=\frac{1}{\gamma}\,(\cC\ofx p+a\ofx),\\
\dot{p}_i&=-\frac{1}{2\gamma}\,p^\tT\,\frac{\d\cC\ofx}{\d x^i}\, p-\frac{1}{\gamma}\,p^\tT\,\frac{\d a\ofx}{\d x^i}\, -\frac{\kappa}{2}\,a\ofx^\tT\,\frac{\d\cC\ofx^{-1}}{\d x^i}\,a\ofx-\kappa a\ofx^\tT\cC\ofx^{-1}\,\frac{\d a\ofx}{\d x^i}\,,\\
\end{split}
\end{equation}
for $i=1,\ldots,n$. In Section \ref{sec:numSec} we shall describe an efficient algorithm to solve these equations numerically.

It is now easy to write down the solution to the Hamilton-Jacobi equation. Indeed, the integral
\begin{equation}\label{eq:s0Expl}
\begin{split}
S^0(t,x\oft)&=-\int_t^T\Big(p\ofs^\tT dx\ofs-H(x\ofs,p\ofs)ds\Big)\\
&=-\int_t^T\Big(\frac{1}{2\gamma}\,p\ofs^\tT\cC(x\ofs)p\ofs-V(x\ofs)\Big)ds
\end{split}
\end{equation}
defines the solution to the Hamilton-Jacobi equation along the characteristic $x\ofs,p\ofs$. In order to find the solution $S^0(t,x)$, for each $x\in\cU$, we eliminate $y$ by inverting the function $y\to x\oft$. Specifically, we write the solution $x\oft$ in the form
\begin{equation}\label{eq:phiDef}
\begin{split}
x\oft&=x(t,y)\\
&=\Phi_t(y),
\end{split}
\end{equation}
which emphasizes the dependence of the trajectory on the terminal value. We have suppressed the terminal value for $p$ as it is always required to be zero. Then, for each $t<T$ is a diffeomorphism of $\cU$. We set
\begin{equation}\label{eq:solHjEq}
S^0(t,x)=S^0(t,x(t,\Phi^{-1}_t(x))).
\end{equation}
This is the desired solution to the Hamilton-Jacobi equation.

The second equation in \eqref{eq:wkbHierarchy} is an inhomogeneous linear first order partial differential equations and it can be readily solved by means of the method of characteristics. Note that, on a characteristic $(x\ofs, p\ofs)$,
\begin{equation*}
\dot{x}\ofs=\frac{1}{\gamma}\,(\cC(x\ofs) a(x\ofs)+\nabla S^0(s,x\ofs)).
\end{equation*}
Therefore, along $x\ofs$, the equation for $S^1$ can be written as an ordinary differential equation,
\begin{equation}
\frac{d}{ds}\,S^1(s,x\ofs)+\frac{1}{2}\,\tr\big(\cC(x\ofs)\nabla^2 S^0(s,x\ofs))\big)=0,
\end{equation}
and thus its solution reads:
\begin{equation}
S^1(t,x\oft)=\frac{1}{2}\,\int_t^T\tr\big(\cC(x\ofs)\nabla^2 S^0(s,x\ofs)\big)ds.
\end{equation}
In analogy with \eqref{eq:phiDef}, we write
\begin{equation}\label{eq:psiDef}
\begin{split}
p\oft&=p(t,y)\\
&=\Psi_t(y).
\end{split}
\end{equation}
Then
\begin{equation}
\begin{split}
p(t,x)&\triangleq\nabla_x S^0(t,x)\\
&=\Psi_t(\Phi^{-1}_t(x)),
\end{split}
\end{equation}
and we can write $S^1(t,x)$ as
\begin{equation}\label{eq:s1Sol}
S^1(t,x)=\frac{1}{2}\,\int_t^T\tr\big(\cC(x\ofs)\nabla p(s,x\ofs)\big)ds.
\end{equation}

Likewise, the solution to the third equation in \eqref{eq:wkbHierarchy} can be written explicitly as
\begin{equation}\label{eq:s2Sol}
\begin{split}
S^2(t,x\oft)&=\frac{1}{2\gamma}\,\int_t^T(\nabla S^1(s,x\ofs))^\tT\cC(x\ofs)\nabla S^1(s,x\ofs)ds\\
&+\frac{1}{2}\,\int_t^T \tr\big(\cC(x\ofs)\nabla^2 S^1(s,x\ofs)\big)ds,
\end{split}
\end{equation}
along a characteristic $x\ofs$. Note that the solution requires knowledge of $S^1$, which in turn requires knowledge of $S^0$. We can continue this process to solve for $S^n$, with the understanding that the complexity of the solution increases in $n$.

Let us now consider the case of $\kappa=-1$, which corresponds to the CARA utility function. This turns out to be a Hamilton's canonical equations read:
\begin{equation*}
\begin{split}
\dot{x}\ofs&=0,\\
\dot{p}\ofs&=\nabla V(x\ofs),\\
\dot{z}\ofs&=V(x\ofs),
\end{split}
\end{equation*}
where $V\ofx$ is defined by \eqref{eq:vDef}. Consequently,
\begin{equation*}
\begin{split}
x\ofs&=y,\\
p\ofs&=-\nabla V(y)(T-s),\\
z\ofs&=-V(y)(T-s)
\end{split}
\end{equation*}
and thus the solution to the Hamilton-Jacobi equation reads
\begin{equation}
S^0(t,x)=-V(x)(T-t).
\end{equation}

Furthermore, we find easily that
\begin{equation*}
S^1(t,x)=\frac{1}{4}\,\tr\big(\cC\ofx\nabla^2 V\ofx\big)(T-t)^2,
\end{equation*}
and
\begin{equation*}
S^2(t,x)=\frac{1}{24}\,\big(\tr(\cC\ofx\nabla^2)^2\big) V\ofx(T-t)^3
\end{equation*}
are the solutions to the second and third equations of the WKB hierarchy, respectively.

\section{\label{sec:exampSec}Generalized Merton portfolio models}

In this section we illustrate the expansion method developed above with a class of portfolio models that are frequently discussed in the literature. Namely, we consider a portfolio of assets whose price dynamics are of the form:
\begin{equation}\label{eq:sepDyn}
\begin{split}
dX^i\oft&=\mu^i(X^i\oft)dt+\sigma^i(X^i\oft) dB_i\oft,\\
X^i\of0&=X^0,
\end{split}
\end{equation}
i.e. the drift $\mu^i(X^i)\in\bR$ and $\sigma^i(X^i)\in\bR$ are functions of $X^i$ only. The Brownian motions $B_i\oft$ above are correlated,
\begin{equation}
\eE[dB_i\oft dB_j\oft]=\rho_{ij}dt.
\end{equation}
This dynamics becomes a special case of \eqref{eq:xDyn}, if we set
\begin{equation}
B\oft=Z\oft L,
\end{equation}
where $Z\oft$ is the standard $n$-dimensional Brownian motion and $L$ is the lower triangular matrix in the Cholesky decomposition, $\rho=L^\tT L$. The model specification \eqref{eq:sepDyn} is natural if we believe that the return and volatility of an asset are local functions of that asset's price only, while the dependence between the assets is a function of the portfolio. Models of this type generalize dynamic portfolio models introduced and studied by Merton \cite{M69}, \cite{M71}.

The covariance matrix and its inverse in this model are given by
\begin{equation}\label{eq:locMetTens}
\begin{split}
\cC_{ij}\ofx&=\rho_{ij}\sigma^i(x^i)\sigma^j(x^j)\,,\\
(\cC\ofx^{-1})_{ij}&=\frac{(\rho^{-1})_{ij}}{\sigma^i(x^i)\sigma^j(x^j)}\,.
\end{split}
\end{equation}
Hence,
\begin{equation}
V\ofx=\frac{\kappa}{2}\,\mu\ofx^\tT\cC\ofx^{-1}\mu\ofx,
\end{equation}
and consequently,
\begin{equation}
\frac{\d}{\d x^i}\, V\ofx=\kappa\Big(\frac{d\mu^i(x^i)}{d x^i}-\mu^i(x^i)\,\frac{d\log\sigma^i(x^i)}{d x^i}\Big)\,\sum\nolimits_j(\cC\ofx^{-1})_{ij}\mu^j(x^j).
\end{equation}
Hence, Hamilton's equations for this model read:
\begin{equation}
\begin{split}
\dot{x}^i&=\frac{1}{\gamma}\Big((\cC\ofx p)_i+\mu^i(x^i)\Big),\\
\dot{p}_i&=-\frac{1}{\gamma}\,p_i\,\Big(\frac{d\log\sigma^i(x^i)}{dx^i}(\cC\ofx p)_i
+\frac{d\mu^i(x^i)}{dx^i}\Big)\\
&\quad-\kappa\Big(\frac{d\mu^i(x^i)}{d x^i}-\mu^i(x^i)\,\frac{d\log\sigma^i(x^i)}{d x^i}\Big)\,(\cC\ofx^{-1}\mu(x))_i.
\end{split}
\end{equation}
These equations are subject to the terminal value conditions:
\begin{equation}\label{eq:termValCond}
\begin{split}
x(T)=y,\\
p(T)=0.
\end{split}
\end{equation}

Let us consider two explicit examples: (i) a portfolio of lognormal assets, and (ii) a portfolio of mean reverting normal assets. A special feature of these examples is that the semiclassical approximations are, in fact, the exact solutions.

\noindent
\underline{\textbf{Optimal control of the multivariate lognormal process.}} As a special case of the model above, we consider a portfolio of $n$ assets each of which follows the lognormal process, i.e.
\begin{equation}\label{eq:lnDyn}
\begin{split}
\mu^i(x^i)&=\mu_i x^i,\\
\sigma^i(x^i)&=\sigma_i x^i,
\end{split}
\end{equation}
where $\mu_i$ and $\sigma_i$ are constant coefficients referred to as the return and lognormal volatility, respectively. This is essentially the original Merton model.

Note that, in this model,
\begin{equation}
V\ofx=\frac{\kappa}{2}\,\mu^\tT C^{-1}\mu
\end{equation}
is constant. Here we have set $C_{ij}=\rho_{ij}\sigma_i\sigma_j$, for $1\leq i,j\leq n$. Hamilton's equations read:
\begin{equation}
\begin{split}
\dot{x}^i&=\frac{1}{\gamma}\Big((\cC\ofx p)_i+\mu_i x^i\Big),\\
\dot{p}_i&=-\frac{1}{\gamma}\,p_i\,\Big(\frac{1}{x^i}\,(\cC\ofx p)_i+\mu_i\Big).
\end{split}
\end{equation}
Since $p(T)=0$, the second of the equations has the unique solution
\begin{equation}
p_i\ofs=0.
\end{equation}
Hence,
\begin{equation}
x^i\ofs=y^i e^{-(\mu_i/\gamma)(T-t)}
\end{equation}
is the unique solution to the first equation subject to the terminal condition $p(T)=y$. These are the characteristics of the Hamilton-Jacobi equation.

This implies that
\begin{equation}
S^0(t,x)=\frac{\kappa}{2}\,\mu^\tT C^{-1}\mu(T-t),
\end{equation}
and
\begin{equation}
S^j(t,x)=0,
\end{equation}
for all $j\geq 1$. Consequently,
\begin{equation}
\begin{split}
\varphi^\ast_i&=\frac{1}{A_U(w)}\,(\cC\ofx^{-1}\mu\ofx)_i\\
&=\frac{1}{A_U(w)}\,\frac{(C^{-1}\mu)_i}{x^j}\,.
\end{split}
\end{equation}
The semiclassical solution is exact and it coincides with Merton's original solution.

\noindent
\underline{\textbf{Optimal control of the multivariate Ornstein-Uhlenbeck process.}} Another tractable portfolio model arises as follows. We consider a portfolio of $n$ assets each of which follows the Ornstein-Uhlenbeck process, i.e. $\mu^i\ofx=\lambda_i(\bar{\mu}^i-x^i)$, and $\sigma^i(x^i)=\sigma_i$, where $\lambda_i$ is the speed of mean reversion of asset $i$, $\bar{\mu}^i$ is its mean reversion level, and $\sigma_i$ is its instantaneous volatility. Note that in this model $V\ofx$ is quadratic,
\begin{equation}
V\ofx=\frac{\kappa}{2}\,(\bar{\mu}-x)^\tT\Lambda \cC^{-1}\Lambda(\bar{\mu}-x),
\end{equation}
where $\Lambda\in\matn$ is the diagonal matrix with entries $\lambda_i$, $i=1,\ldots, n$.

As a result, Hamilton's equations can be solved in closed form. Indeed, we find that they form a linear system:
\begin{equation}\label{eq:linHam}
\frac{d}{dt}
\begin{pmatrix}
x\\
p
\end{pmatrix}
=A
\begin{pmatrix}
x\\
p
\end{pmatrix}
+m,
\end{equation}
where
\begin{equation}
\begin{split}
A&=
\begin{pmatrix}
-\gamma^{-1}\Lambda& \gamma^{-1}\cC\\
-\kappa\Lambda\cC^{-1}\Lambda& \gamma^{-1}\Lambda
\end{pmatrix},\\
m&=
\begin{pmatrix}
\gamma^{-1}\Lambda\bar{\mu}\\
\kappa\Lambda\cC^{-1}\Lambda\bar{\mu}
\end{pmatrix}.
\end{split}
\end{equation}
The solution to the system \eqref{eq:linHam} subject to the terminal conditions $p(T)=0$ and $x(T)=y$ reads:
\begin{equation}
\begin{pmatrix}
x\ofs\\
p\ofs
\end{pmatrix}
=e^{-(T-s)A}\Bigg(
\begin{pmatrix}
y\\
0
\end{pmatrix}
+A^{-1}m\Bigg)-A^{-1}m,
\end{equation}
where the exponential denotes the matrix exponential function. These are the characteristics of the Hamilton-Jacobi equation.

This representation allows us to explicitly construct the maps $\Phi_t$ and $\Psi_t$ in \eqref{eq:phiDef} and \eqref{eq:psiDef}, respectively. Indeed, they are linear in $y$ and, consequently, $\Phi_t^{-1}$ and $\Psi_t^{-1}$ are linear functions as well. As a consequence of \eqref{eq:s0Expl}, $S^0(t,x)$ is an explicitly computable quadratic function of $x$. Since in the current model $\cC$ is independent of $x$, formula \eqref{eq:s1Sol} implies that $S^1(t,x)$ is non-zero but independent of $x$. Inspection of the WKB hierarchy shows immediately that
\begin{equation}
S^j(t,x)=0,
\end{equation}
for all $j\geq 2$. Consequently, we obtain the following formula for the optimal control:
\begin{equation}
\varphi^\ast=\frac{1}{A_U(w)}\,\big(\cC^{-1}\Lambda(\bar{\mu}-x)+\nabla S^0(t,x)\big),
\end{equation}
with no further corrections. As in the case of the lognormal model, this semiclassical solution turns out to be exact.

\section{\label{sec:numSec}Numerical implementation of the solution}

Interesting cases for which the Hamilton-Jacobi equation \eqref{eq:hjEqu} admits a closed form solution are scarce. In fact, even in the case of constant covariance matrix, \eqref{eq:hjEqu} cannot, in general, be solved in closed form. In this section we discuss a numerical solution method for the Hamilton-Jacobi equation, and indeed (at least in principle) the entire WKB hierarchy, that is efficient and accurate for systems with a relatively large ($\lessapprox 200$) number of degrees of freedom.

Let
\begin{equation}
\begin{split}
x\oft&=\Phi_t(y),\\
p\oft&=\Psi_t(y),
\end{split}
\end{equation}
denote the solution to Hamilton's equations \eqref{eq:charEqHJ} with terminal condition \eqref{eq:termValCharEqHJ}. Our goal is to compute $S^0(t,x)$ and $S^1(t,x)$ for all $0\leq t\leq T$ and $x\in\cU$. This amounts to an effective numerical implementation of the solutions constructed in Section \ref{sec:solHier} by means of the method of characteristics. We proceed in the following steps.
\begin{itemize}
\item[]{\emph{Step 1}. For a given terminal value $y$ and each $t<T$, compute $x\oft=\Phi_t(y)$ and $p\oft=\Psi_t(y)$.}
\item[]{\emph{Step 2}. Given $x\in\cU$ and $t<T$ find $y$ such that $x\oft=x$. This is equivalent to inverting the function $y\to\Phi_t\ofy$.}
\item[]{\emph{Step 3}. Given $x\in\cU$ and $t<T$, compute $S^0(t,x)$.}
\item[]{\emph{Step 4}. Given $x\in\cU$ and $t<T$, compute $S^1(t,x)$.}
\end{itemize}
We shall now describe these steps in detail.

\noindent
\emph{Step 1}. In order to construct the pair $(\Phi_t, \Psi_t)$ we use the St\"ormer-Verlet / leapfrog method of integrating Hamilton's equations \cite{HLW03}, \cite{LR04}. Other popular numerical methods, such as Euler's method or the Runge-Kutta method, tend to perform poorly when applied to a Hamiltonian system. This can be traced to the fact that these methods do not respect the underlying symplectic structure, and, in particular, do not preserve the volume in the phase space of the system. The leapfrog method is an ingenious way of discretizing a Hamiltonian system, so that it defines a symplectic map. As a additional bonus, the St\"ormer-Verlet / leapfrog scheme is order $h^2$ accurate.

Specifically, we discretize the time range $[t,T]$,
\begin{equation}\label{eq:discT}
t_k=t+kh,\text{ if } k=0,1,\ldots,N,
\end{equation}
where the time step $h=(T-t)/N$ is chosen suitably. We replace the continuous time Hamiltonian system \eqref{eq:charEqHJ} by a discrete time dynamical system, and let $\hat{x}_k$ and $\hat{p}_k$ denote the approximate values of $x(t_k)$ and $p(t_k)$, respectively. We require that $\hat{x}_k$ and $\hat{p}_k$ follow the numerical scheme:
\begin{equation}\label{eq:genLeap}
\begin{split}
\hat{p}_{k-\frac{1}{2}}&=\hat{p}_{k}+\frac{h}{2}\,\nabla_x H(\hat{x}_k,\hat{p}_{k-\frac{1}{2}})\,,\\
\hat{x}_{k-1}&=\hat{x}_k-\frac{h}{2}\,\big(\nabla_p H(\hat{x}_k,\hat{p}_{k-\frac{1}{2}})+\nabla_p H(\hat{x}_{k-1},\hat{p}_{k-\frac{1}{2}})\big),\\
\hat{p}_{k-1}&=\hat{p}_{k-\frac{1}{2}}+\frac{h}{2}\,\nabla_x H(\hat{x}_{k-1},\hat{p}_{k-\frac{1}{2}}),
\end{split}
\end{equation}
where we have introduced half intervals values $\hat{p}_{k-\frac{1}{2}}$. The presence of these intermediate values of the momentum is the crux of the leapfrog method and it assures that the scheme is symplectic. Notice that the first and second equations in \eqref{eq:genLeap} are implicit in $\hat{p}_{k-\frac{1}{2}}$ and $\hat{x}_{k-1}$, respectively.

Calculating the derivatives yields
\begin{equation}\label{eq:LeapConcIF}
\begin{split}
\hat{p}_{k-\frac{1}{2}}&=\hat{p}_k+\frac{h}{2\gamma}\,\nabla a(\hat{x}_k)\hat{p}_{k-\frac{1}{2}}+\frac{h}{2\gamma}\,\hat{p}_{k-\frac{1}{2}}^\tT \nabla \cC(\hat{x}_k)\hat{p}_{k-\frac{1}{2}}+\frac{h}{2}\,\nabla V(\hat{x}_k),\\
\hat{x}_{k-1}&=\hat{x}_k-\frac{h}{2\gamma}\,\big(\cC(\hat{x}_k)
+\cC(\hat{x}_{k-1})\big)\hat{p}_{k-\frac{1}{2}}-\frac{h}{2\gamma}\,\big(a(\hat{x}_k)+a(\hat{x}_{k-1})\big),\\
\hat{p}_{k-1}&=\hat{p}_{k-\frac{1}{2}}+\frac{h}{2\gamma}\,\nabla a(\hat{x}_{k-1})\hat{p}_{k-\frac{1}{2}}+\frac{h}{2\gamma}\,\hat{p}_{k-\frac{1}{2}}^\tT \nabla \cC(\hat{x}_{k-1})\hat{p}_{k-\frac{1}{2}}+\frac{h}{2}\,\nabla V(\hat{x}_{k-1}).
\end{split}
\end{equation}
This system is subject to the terminal condition:
\begin{equation}
\begin{split}
\hat{x}_N&=y,\\
\hat{p}_N&=0.
\end{split}
\end{equation}
Note that the first two relations in \eqref{eq:LeapConcIF} cannot, in general, be solved explicitly for $\hat{p}_{k-\frac{1}{2}}$ and $\hat{x}_{k-1}$, respectively, and thus they need to be solved numerically. This can be efficiently accomplished, for example by means of Newton's method with the initial guess $\hat{p}_{k-\frac{1}{2}}=\hat{p}_k$ and $\hat{x}_{k-1}=\hat{x}_k$. Indeed, in practice, a few iterations of Newton's method yield a very accurate solution.

Solving this system yields an approximate flow map $\hat{\Phi}_t$. Throughout the reminder of this section we shall suppress the hat over $x,p$, etc. keeping in mind that all the quantities are numerical approximations to the true values.

\noindent
\emph{Step 2}. In order to carry out the next step, we develop an algorithm for inverting the flow map $\Phi_t:\cU\to\cU$ defined above. From the existence theory of ordinary differential equations, $x$ and $p$ depend smoothly on the terminal value $y$. Hence, the sensitivities $\nabla_y\Phi$ and $\nabla_y\Psi$ satisfy the following system of equations:
\begin{equation}\label{eq:varEqs}
\begin{split}
\frac{d}{dt}\,\nabla_y\Phi&=\nabla^2_{px}H\,\nabla_y\Phi+\nabla^2_{pp}H\,\nabla_y\Psi,\\
\frac{d}{dt}\,\nabla_y\Psi&=-\nabla^2_{xx}H\,\nabla_y\Phi-\nabla^2_{xp}H\,\,\nabla_y\Psi,
\end{split}
\end{equation}
subject to the terminal condition
\begin{equation}
\begin{split}
\nabla_y\Phi(T,y)&=\id,\\
\nabla_y\Psi(T,y)&=0.
\end{split}
\end{equation}
Equations of this type are known as variational equations, see e.g. \cite{G27}.

Consider now an approximation of the variational system \eqref{eq:varEqs}, in which the second derivatives of $H$ are evaluated at the constant trajectory $(x\oft,p\oft)=(y,0)$. We thus obtain the following linear system with constant coefficients:
\begin{equation}
\begin{split}
\dot{F}&=Q\ofy F+R\ofy G,\\
\dot{G}&=-U\ofy F-Q\ofy G,
\end{split}
\end{equation}
where the matrices $Q$, $R$, and $U$ are given explicitly by
\begin{equation}
\begin{split}
Q\ofy&=(\kappa+1)\nabla a\ofy^\tT,\\
R\ofy&=(\kappa+1)\cC\ofy,\\
U\ofy&=\nabla^2 V\ofy.
\end{split}
\end{equation}
Note that $F_t\ofy\equiv F(t,y)$ is an approximation to $\nabla\Phi_t\ofy$, the gradient of the function $y\to\Phi_t\ofy$. This linear system can be written in a more compact form as
\begin{equation}
\frac{d}{dt}
\begin{pmatrix}
F\\
G\\
\end{pmatrix}
=M\ofy
\begin{pmatrix}
F\\
G\\
\end{pmatrix},
\end{equation}
where
\begin{equation}
M\ofy=
\begin{pmatrix}
Q\ofy & R\ofy\\
-U\ofy & -Q\ofy\\
\end{pmatrix},
\end{equation}
subject to the terminal condition
\begin{equation}
\begin{pmatrix}
F_T\ofy\\
G_T\ofy\\
\end{pmatrix}
=
\begin{pmatrix}
\id\\
0\\
\end{pmatrix}.
\end{equation}
This problem has a unique solution, namely
\begin{equation}\label{eq:linSol}
\begin{pmatrix}
F_t\ofy\\
G_t\ofy\\
\end{pmatrix}
=e^{-(T-t)M\ofy}
\begin{pmatrix}
\id\\
0\\
\end{pmatrix},
\end{equation}
where, as before, the exponential denotes the matrix exponential function. This solution can readily be implemented in computer code \cite{GL}.

Now, our next goal is to solve for $y$ the equation
\begin{equation}
\Phi_t\ofy-x=0.
\end{equation}
To this end, we use a Newton-type method. Finding the gradient $\nabla\Phi_t\ofy$ is computationally very expensive, and it may be susceptible to numerical inaccuracies. Fortunately, for convergence purposes, it is sufficient to approximate it with $F_t\ofy$, which we have just computed explicitly. In the Appendix we justify this procedure, by proving that it converges for $T$ sufficiently small. The following pseudocode implements this search algorithm:
\begin{equation}
\begin{split}
&eps\gets 10^{-13}\\
&y\gets x\\
&err\gets1.0\\
&\texttt{while}(err>eps)\\
&\qquad z\gets y-F_t\ofy^{-1}(\Phi_t\ofy-x)\\
&\qquad err\gets\|z-y\|\\
&\qquad y\gets z
\end{split}
\end{equation}
The norm $\|\cdot\|$ above denotes the usual Euclidean norm in $\bR^n$.

\noindent
\emph{Step 3}. We are now ready to compute the value of $S^0(t,x)$. In Step 3 we have found $y=\Phi_t^{-1}(x)$. Using the algorithm explained in Step 1, we construct the discrete trajectory $(x_{t_k},p_{t_k})$. We write the integral \eqref{eq:s0Expl} as a sum of integrals over the segments $[t_k,t_{k+1}]$,
\begin{equation}
S^0(t,x)=\sum_{k=0}^{N-1}\,I^0_{t_k,t_{k+1}}.
\end{equation}
We denote the integrand in \eqref{eq:s0Expl} by $L(x\ofs,p\ofs)$, and calculate each of the subintegrals according to Simpson's rule:
\begin{equation}
\begin{split}
I^0_{a,b}&=\int_a^b L(x(s),p(s))ds\\
&\approx\frac{1}{6}\big(L(x(a),p(a))+4L(x(m),p(m))+L(x(b),p(b))\big)(b-a),
\end{split}
\end{equation}
where $m=\tfrac{a+b}{2}$ is the midpoint between $a$ and $b$.

\noindent
\emph{Step 4}. We break up the integral in \eqref{eq:s1Sol} into the sum of integrals over $[t_k,t_{k+1}]$,
\begin{equation}
S^1(t,x)=\sum_{k=0}^{N-1}\,I^1_{t_k,t_{k+1}}.
\end{equation}
Reusing the discrete trajectory $(x_{t_k},p_{t_k})$ calculated in Step 3, we compute each subintegral using Simpson's rule:
\begin{equation}
\begin{split}
I^1_{a,b}&=\frac{1}{2}\int_a^b\tr\big(\cC(x\ofs)\nabla p(s,x\ofs)\big)ds\\
&\approx\frac{1}{6}\,\tr\big(\cC(x(a))\nabla p(a,x(a))+4\cC(x(m))\nabla p(m,x(m))+\cC(x(b))\nabla p(b,x(b))\big)(b-a).
\end{split}
\end{equation}
The first order partial derivatives in $\nabla p$ in the expression above are calculated as central finite differences:
\begin{equation}
\frac{\d}{\d x^j}\, p_j(t_k,x(t_k))\approx\frac{p_{t_{k+1}}-p_{t_{k-1}}}{2(x_{t_{k+1}}-x_{t_{k-1}})}\,.
\end{equation}

\appendix

\section{\label{sec:UtilityFunctions}The HARA family of utility functions}

We let $U\ofv$ denote a utility function, i.e. a twice differentiable concave function. Recall that the \emph{absolute risk aversion} coefficient associated with $U\ofv$ is defined by
\begin{equation}
A_U\ofv=-\frac{U^{\prime\prime}\ofv}{U^{\prime}\ofv}\;,
\label{eq:ara}
\end{equation}
while the \emph{relative risk aversion} coefficient is given by
\begin{equation}
\begin{split}
R_U\ofv&=-\frac{vU^{\prime\prime}\ofv}{U^{\prime}\ofv}\\
&=vA_U\ofv\;.
\end{split}
\label{eq:rra}
\end{equation}

In this paper we consider the following four utility functions:
\begin{itemize}
\item[(1)]{The \emph{hyperbolic absolute risk aversion} (HARA) utility,
\begin{equation}
U_{\rm{HARA}}\left(v;a,b,\gamma\right)=\frac{\gamma}{1-\gamma}\left(a+\frac{b}{\gamma}\,v \right)^{1-\gamma}\;.
\label{eq:HARA}
\end{equation}}
\item[(2)]{The \emph{constant relative risk aversion} (CRRA) utility,
\begin{equation}
U_{\rm{CRRA}}\left(v;\gamma\right)=\frac{v^{1-\gamma}}{1-\gamma}\;.
\label{eq:CRRA}
\end{equation}
Note that $\gamma$ is the (constant) relative risk aversion coefficient associated with this utility function, $R_U\ofv=\gamma$.}
\item[(3)]{The \emph{constant absolute risk aversion} (CARA) utility,
\begin{equation}
U_{\rm{CARA}}\left(v;\gamma\right)=-\frac{e^{-\gamma v}}{\gamma}\;,
\label{eq:CARA}
\end{equation}
Note that $\gamma$ is the (constant) absolute risk aversion coefficient associated with this utility function, $A_U\ofv=\gamma$.}
\item[(4)]{The \emph{logarithmic} (Bernoulli) utility,
\begin{equation}
U_{\rm{LOG}}\ofv=\log\ofv.
\end{equation}}
\end{itemize}

It is well known that the HARA utility includes the CRRA, CARA, and logarithmic utility functions as limit cases. Indeed,
\begin{equation*}
\begin{split}
U_{\rm{CRRA}}\left(v;\gamma\right)&=U_{\rm{HARA}}\left(v;0,\gamma^{-\gamma/\left(1-\gamma\right)},\gamma\right),\\
U_{\rm{CARA}}\left(v;\gamma\right)&=\frac{1}{\gamma}\;\lim_{c\to\infty}U_{\rm{HARA}}\left(v;1,\gamma,c\right),\\
U_{\rm{LOG}}\ofv&=\lim_{\gamma\to 1}U_{\rm{CRRA}}\left(v;\gamma\right).\\
\end{split}
\end{equation*}

The following proposition is used in Section \ref{sec:hjbEq}.
\begin{proposition}\label{thm:haraProps}
Let $U$ be a twice differentiable function. The ratio
\begin{equation}
-\frac{U'\ofv^2}{U''\ofv U\ofv}
\end{equation}
is a constant if and only if $U$ is a HARA utility function. In this case, its value $\kappa$ is given by
\begin{equation}
\kappa=
\begin{cases}
\left(1-\gamma\right)/\gamma,& \text{ for the HARA and CRRA utilities,}\\
-1,& \text{ for the CARA utility,}\\
0,& \text { for the log utility.}
\end{cases}\label{eq:kappaDef}
\end{equation}
\end{proposition}
The proof of this proposition is a straightforward calculation and we omit it.

\section{\label{sec:convSec}Convergence of the modified Newton method}

The purpose of this Appendix is to prove that the Newton-type method described in Section \ref{sec:numSec} converges, at least in the case if the time horizon $T$ is sufficiently short. Our proof uses familiar techniques of numerical analysis \cite{SB02} and systems of ordinary differential equations \cite{CL55}.

We first state the following general fact.
\begin{proposition}\label{thm:newtMeth}{Let $\cB\subset\bR^n$ be a compact set, and let $h:\cB\to\cB$ be a twice continuously differentiable function. Assume that $h$ has a unique simple zero $y^*\in\cB$,
\begin{equation}
\begin{split}
h(y^*)&=0,\\
\nabla h(y^*)&\neq 0.
\end{split}
\end{equation}
Let $F:\cB\to\matn$ be a continuously differentiable function such that $F\ofy^{-1}$ exists for all $y\in\cB$, and the following two conditions are satisfied. There is a $0<\delta<1$, such that
\begin{equation}\label{eq:cond1}
\|\id-F\ofy^{-1}\nabla h\ofy\|\leq\delta/2,
\end{equation}
and
\begin{equation}\label{eq:cond2}
\|\nabla\big(F\ofy^{-1}\big)h\ofy\|\leq\delta/2,
\end{equation}
for all $y\in\cB$. Then the map
\begin{equation}
f\ofy=y-F\ofy^{-1}h\ofy
\end{equation}
is a contraction of $\cB$ into itself.}
\end{proposition}
\begin{proof}
We verify easily that conditions \eqref{eq:cond1} and \eqref{eq:cond2} imply that $\|\nabla f\ofy\|\leq 1-\delta$, uniformly in $y\in\cB$. Hence, $f$ is a contraction.
\end{proof}

As a consequence of this proposition and the contraction principle, the sequence
\begin{equation}\label{eq:nrSeq}
\begin{split}
y_1&=f(y_0)\\
&=y_0-F(y_0)^{-1}h(y_0),\\
y_2&=f(y_1)\\
&=y_1-F(y_1)^{-1}h(y_1),\\
&\ldots,
\end{split}
\end{equation}
where $y_0\in\cB$ is arbitrary, converges to $y^*$.

Next, we shall show that the above proposition applies to our specific situation. The sequence \eqref{eq:nrSeq} will furnish the modified Newton method used in Section \ref{sec:numSec}.
\begin{proposition}{Assume that the Hamiltonian $H(x,p)$ is three times continuously differentiable, and let $h\ofy=\Phi_t\ofy-x$ and $F\ofy=F_t\ofx$. Then, there are a $T>0$ and a compact set $\cB\ni y$ such that conditions \eqref{eq:cond1} and \eqref{eq:cond2} of Proposition \ref{thm:newtMeth} are satisfied.
}
\end{proposition}
\begin{proof}
By the general theory of ordinary differential equations (see e.g. \cite{CL55}), we note first that under our assumptions, there is a $T> 0$ such that $\Phi_t\ofy$ and $\Psi_t\ofy$ are unique solutions of the terminal value problem \eqref{eq:charEqHJ} - \eqref{eq:termValCharEqHJ}, and they have continuous derivatives with respect to $y$. Furthermore, it is clear from \eqref{eq:linSol} that $F_t\ofy$ is nonzero for all $y$, and
\begin{equation}\label{eq:bound1}
\max_{y\in\cB}\|F_t\ofy^{-1}\|<\infty,
\end{equation}
\begin{equation}\label{eq:bound2}
\max_{y\in\cB}\|\nabla\big(F_t\ofy^{-1}\big)\|<\infty.
\end{equation}

Now, in order to prove \eqref{eq:cond1}, it is sufficient to prove that given an $\eta>0$, there is a $T>0$ such that
\begin{equation}\label{eq:etaIneq}
\|F\ofy-\nabla h\ofy\|\leq\eta,
\end{equation}
for all $t\leq T$. Indeed,
\begin{equation*}
\begin{split}
\|\id-F\ofy^{-1}\nabla h\ofy\|&=\|F\ofy^{-1}(F\ofy-\nabla h\ofy)\|\\
&\leq\max_{y\in\cB}\|F\ofy^{-1}\|\|F\ofy-\nabla h\ofy\|\\
&\leq\mathrm{const}\times\eta.
\end{split}
\end{equation*}
In order to prove \eqref{eq:etaIneq}, we set
\begin{equation*}
D_t(y)=
\begin{pmatrix}
\nabla_y\Phi_t\ofy\\
\nabla_y\Psi_t\ofy
\end{pmatrix}
-
\begin{pmatrix}
F_t\ofy\\
G_t\ofy
\end{pmatrix},
\end{equation*}
and
\begin{equation*}
N_t(y)=
\begin{pmatrix}
\nabla^2_{xp}H(\Phi_t\ofy,\Psi_t\ofy)&\nabla^2_{xp}H(\Phi_t\ofy,\Psi_t\ofy)\\
\nabla^2_{xp}H(\Phi_t\ofy,\Psi_t\ofy)&\nabla^2_{xp}H(\Phi_t\ofy,\Psi_t\ofy)
\end{pmatrix}.
\end{equation*}
Then $D_t\ofy$ satisfies the following system of differential equations:
\begin{equation}
\begin{split}
\dot{D}_t\ofy&=N_t\ofy
\begin{pmatrix}
\nabla_y\Phi_t\ofy\\
\nabla_y\Psi_t\ofy
\end{pmatrix}
-M\ofy
\begin{pmatrix}
F_t\ofy\\
G_t\ofy
\end{pmatrix}\\
&=M\ofy D_t\ofy+E_t\ofy,
\end{split}
\end{equation}
where
\begin{equation*}
E_t\ofy=(N_t\ofy-M\ofy)
\begin{pmatrix}
\nabla_y\Phi_t\ofy\\
\nabla_y\Psi_t\ofy
\end{pmatrix},
\end{equation*}
subject to the terminal condition $D_T\ofy=\id$. Hence
\begin{equation}
D_t\ofy=\int_T^t e^{(t-s)M\ofy}E_s\ofy ds.
\end{equation}
As a consequence
\begin{equation*}
\begin{split}
\|D_t\ofy\|&\leq\int_T^t\|e^{(t-s)M\ofy}E_s\ofy\|\,ds\\
&\leq\mathrm{const}\max_{t\leq s\leq T}\|E_s\ofy\|\,T\\
&\leq\mathrm{const}\max_{(x,p)\in\cB\times\cP}\|\nabla^3 H(x,p)\|\,T,
\end{split}
\end{equation*}
where the constant is independent of $T$. The set $\cP$ is a bounded subset of $\bR^n$ which contains the trajectory of $p\oft$, for $0\leq t\leq T$. Since the maximum above is finite, we conclude that $\|D_t\ofy\|\leq\mathrm{const}\,T$.

Condition \eqref{eq:cond2} is a consequence of \eqref{eq:bound2} and the fact that we can choose $\cB$ sufficiently small so that $\|h\ofy\|$ is less than any given number.
\end{proof}

\end{document}